\documentclass{article}
\usepackage{ifthen}
\usepackage{mdwlist}
\usepackage{fullpage}
\usepackage{amsmath,amsfonts,amsthm,mathrsfs}
\usepackage{bm}
\usepackage{enumitem}
\usepackage{graphicx}
\usepackage[ruled,linesnumbered]{algorithm2e}
\usepackage{algorithmic}
\usepackage{eqparbox}

\usepackage[pdftex]{color}

\newtheorem{theorem}{Theorem}

\newtheorem{lemma}{Lemma}

\newtheorem{proposition}{Proposition}

\newtheorem{definition}{Definition}

\newcommand{\MAX}{\mathsf{MAX}}

\newcommand{\junk}[1]{}
\junk{

}

\newcommand{\R}{\ensuremath{\mathbb{R}}} 

\newcommand{\wt}{{w}}
\newcommand{\mtd}{{\mathcal{M}}}
\newcommand{\gset}{{\mathcal{U}}}
\newcommand{\isets}{{\mathcal{I}}}
\newcommand{\distr}{{F}}
\newcommand{\dv}{{b}}
\newcommand{\vctr}[1]{{\mathbf{#1}}}
\newcommand{\expect}[1]{{\mathbb{E} \left[ #1 \right]}}

\newcommand{\opt}{{\mathsf{OPT}}}
\newcommand{\matrank}{{\operatorname{rank}}}
\newcommand{\matspan}{{\operatorname{cl}}}

\newcommand{\iseq}{\vctr{\sigma}}

\newcommand{\icopies}{{\isets^{\rm copies}}}
\newcommand{\mech}{{\mathcal{M}}}
\newcommand{\mcopies}{\mech^{\rm copies}}
\newcommand{\optcopies}{\opt^{\rm copies}}

\title{Matroid Prophet Inequalities}
\author{
Robert Kleinberg\thanks{Department of Computer Science, Cornell University.}
\and S.\ Matthew Weinberg\thanks{Department of Electrical Engineering and Computer Science, MIT.} 
}
\date{}

\begin{document}
\setcounter{page}{0}

\begin{titlepage}

\maketitle

\begin{abstract}
Consider a gambler who observes a sequence of independent, non-negative
random numbers and is allowed to stop the sequence at any time, claiming
a reward equal to the most recent observation.  
The famous prophet inequality of Krengel, Sucheston, and Garling
asserts that a gambler who knows the distribution of each random
variable can achieve at least half as much reward, in expectation,
as a ``prophet'' who knows the sampled values of each random 
variable and can choose the largest one.  We
generalize this result to the setting in which the
gambler and the prophet are allowed to make more than one selection, 
subject to a matroid constraint.  We show that the gambler
can still achieve at least half as much reward as the prophet;
this result is the best possible, since it is known that the
ratio cannot be improved even in the original prophet inequality,
which corresponds to the special case of rank-one matroids.
Generalizing the result still further, we show that under an
intersection of $p$ matroid constraints, the prophet's reward 
exceeds the gambler's by a factor of at most $O(p)$, and this factor is also tight.

Beyond their interest as theorems about pure online algoritms or 
optimal stopping rules, these results also have applications to
mechanism design.  Our results imply improved bounds on the 
ability of sequential posted-price mechanisms to approximate
Bayesian optimal mechanisms in both single-parameter and 
multi-parameter settings.  In particular, our results imply
the first efficiently computable 
constant-factor approximations to the Bayesian
optimal revenue in certain multi-parameter settings.
\end{abstract}

\renewcommand{\thepage}{}

\end{titlepage}

\section{Introduction}

In 1978, Krengel, Sucheston and Garling~\cite{KS78} 
proved a surprising and fundamental
result about the relative power of online and offline algorithms
in Bayesian settings.  
They showed that if
$X_1,X_2,\ldots,X_n$ is a sequence of independent, non-negative, 
real-valued random variables and $\expect{\max_i X_i} < \infty$,
then there exists a stopping rule $\tau$ such that 
\begin{equation} \label{eq:prophet-inequality}
2 \cdot \expect{X_\tau} \geq \expect{\max_i X_i}.
\end{equation}
In other words, if we consider a game in which a player
observes the sequence $X_1,X_2,\ldots,X_n$ and is allowed to
terminate the game at any time, collecting the most recently observed
reward, then a prophet who can foretell the entire sequence
and stop at its maximum value can gain at most twice as much
payoff as a player who must choose the stopping time based
only on the current and past observations.
The inequality~\eqref{eq:prophet-inequality} became the 
first\footnote{More precisely, it was the second
prophet inequality.  
The same inequality with a factor of 4, instead of 2, was
discovered a year earlier by Krengel and Sucheston~\cite{KS77}.}
of many ``prophet inequalities'' in optimal stopping theory.
Expressed in computer science terms, these inequalities
compare the performance of online algorithms versus 
the offline optimum for problems that involve
selecting one or more elements from a random sequence, 
in a Bayesian setting
where the algorithm knows the distribution from which the 
sequence will be sampled whereas the offline optimum 
knows the values of the samples themselves and chooses
among them optimally.  
Not surprisingly, these inequalities have important applications
in the design and analysis of algorithms, especially in 
algorithmic mechanism design, a connection that we discuss
further below.

In this paper, we prove a prophet inequality for matroids,
generalizing the original inequality~\eqref{eq:prophet-inequality}
which corresponds to the special case of rank-one matroids.  
More specifically, we 
analyze the following online selection problem.  
One is given
a matroid whose elements have random weights
sampled independently from (not necessarily identical)
probability distributions on $\R_+$.  An online algorithm,
initialized with knowledge of the matroid structure and of the
distribution of each element's weight, must select an 
independent subset of the matroid by observing the
sampled value of each element (in a fixed, prespecified
order) and making an immediate decision whether or not
to select it before observing the next element.
The algorithm's payoff is defined to be the sum of the
weights of the selected elements.  We prove in this
paper that for every matroid, there is an online algorithm
whose expected payoff is least half of the expected
weight of the maximum-weight basis.  
It is well known that the factor 2 in Krengel,
Sucheston, and Garling's inequality~\eqref{eq:prophet-inequality} 
cannot be improved
(see Section~\ref{sec:lb} for a lower bound example)
and therefore our result for matroids is the best 
possible, even in the rank-one case.

Our algorithm is quite simple.  At its heart lies a new
algorithm for achieving the optimal factor 2
in rank-one matroids: compute a threshold value
$T = \expect{\max_i X_i}/2$ and accept the first 
element whose weight exceeds this threshold.
This is very similar to the algorithm of Samuel-Cahn~\cite{S-C},
which uses a threshold $T$ such that $\Pr(\max_i X_i > T) = \tfrac12$
but is otherwise the same, and which also achieves the 
optimal factor 2.  It is hard to surpass
the elegance of Samuel-Cahn's proof, and indeed our proof,
though short and simple, is not as elegant.  On the other hand, our
algorithm for rank-one matroids has a crucial advantage
over Samuel-Cahn's:
it generalizes to arbitrary matroids without weakening
its approximation factor.  The generalization is as 
follows.  The algorithm pretends that the online selection
process is Phase 1 of a two-phase game; after each $X_i$ has been revealed in Phase 1 and the algorithm has accepted some set $A_1$, Phase 2 begins.
In Phase 2, a new weight will be sampled for every matroid element, independently
of the Phase 1 weights, and the algorithm will
play the role of the prophet on the Phase 2 weights, choosing the max-weight subset $A_2$ such that $A_1 \cup A_2$ is independent.
However, the payoff for choosing an element in Phase 2 is
only half of its weight. When observing element $i$ and
deciding whether to select it, our algorithm can be interpreted as making the choice
that would maximize its expected payoff if Phase 1
were to end immediately after making this decision and Phase 2 were to begin.  Of course, Phase 2 is purely fictional:
it never actually takes place, but it plays a key role in
both the design and the analysis of the algorithm.
Note that this algorithm, specialized to rank-one
matroids, is precisely the one proposed at the start 
of this paragraph:
the expected value of proceeding to Phase 2 without
selecting anything would be $T=\expect{\max_i X_i}/2$, 
hence our
algorithm picks an element if and only if its weight
exceeds $T$.

We next extend our algorithm to the case in which the 
feasibility constraint is given by a matroid intersection
rather than a single matroid.  For intersections of $p$
matroids, we present an online algorithm whose expected
payoff is at least $\tfrac{1}{4p-2}$ times the expected 
maximum weight of a feasible set.  The algorithm is
a natural extension of the one described earlier.
It again imagines a fictional Phase 2 in which new
independent random weights are sampled for all elements
and revealed simultaneously, and the payoff for selecting
an element in Phase 2 equals half of its weight.  This
time, we let $M_2$ denote the max-weight feasible set of Phase 2 elements, designate one of the $p$ matroids uniformly at random, and allow the algorithm to choose any $A_2 \subseteq M_2$ such that $A_1 \cup A_2$ is independent in the designated matroid. Observe that this is in fact a generalization of our algorithm for a single matroid, as enforcing $A_2 \subseteq M_2$ is a vacuous constraint for a single matroid.
In Section~\ref{sec:lb} we show that our result
for matroid intersections is almost tight: we present a lower
bound demonstrating that the ratio $4p-2$ cannot be improved
by more than a constant factor.

As mentioned earlier, Bayesian optimal mechanism design problems
provide a compelling application of prophet inequalities 
in computer science and economics.  In Bayesian optimal mechanism
design, one has a collection of $n$ agents with independent
private types sampled from known distributions, and the goal
is to design a mechanism for allocating resources and charging
prices to the agents, given their reported types, so as to 
maximize the seller's expected revenue in equilibrium.
Chawla et al.~\cite{CHMS} pioneered the study of approximation
guarantees for \emph{sequential posted pricings} (SPMs), a very simple
class of mechanisms in which the seller makes a sequence of
take-it-or-leave-it offers to the agents, with each offer 
specifying an item and a price that the agent must pay in order
to win the item.  Despite their simplicity, sequential posted
pricings were shown in~\cite{CHMS} to approximate the optimal 
revenue in many different settings.  Prophet inequalities
constitute a key technique underlying these results; instead
of directly analyzing the revenue of the SPM, one analyzes
the so-called \emph{virtual values} of the winning bids,
proving via prophet inequalities that the combined expected
virtual value accumulated by the SPM approximates the offline
optimum.  Translating this virtual-value approximation
guarantee into a revenue guarantee is an application
of standard Bayesian mechanism design techniques introduced
by Roger Myerson~\cite{Mye81}.  In the course of developing
these results, Chawla et al.\ prove a type of prophet
inequality for matroids that is of considerable interest
in its own right: they show that if the algorithm is allowed
to specify the order in which the matroid elements are observed,
then it can guarantee an expected payoff at least half as large
as the prophet's.  Our result can be seen as a strengthening
of theirs, achieving the same approximation bound without
allowing the algorithm to reorder the elements.  Unlike our
setting, in which the factor 2 is known to be tight, the
best known lower bound for 
algorithms that may reorder the elements is $\sqrt{\pi/2} \cong
1.25$.

Extending the aforementioned results from single-parameter to
multi-parameter domains, Chawla et al.\ define in~\cite{CHMS} 
a general class of multi-parameter
mechanism design problems called Bayesian multi-parameter unit-demand
(BMUMD).  SPMs in this setting are not truthful but can be modified
to yield mechanisms that approximate the Bayesian optimal
revenue with respect to a weaker solution concept: implementation
in undominated strategies.  A narrower class of mechanisms
called \emph{oblivious posted pricings} (OPMs) yields truthful
mechanisms, but typically with weaker approximation guarantees;
for example, it is not known whether OPMs can yield 
constant-factor approximations to the Bayesian optimal
revenue in matroid settings, except for special cases such
as graphic matroids.  Without resolving this question,
our results lead to an equally strong positive result for BMUMD: 
truthful mechanisms that 2-approximate the Bayesian optimal revenue
in matroid settings and $(4p-2)$-approximate it in settings
defined by an intersection of $p$ matroid constraints.

\subsection{Related work}
\label{sec:relwork}
The genesis of prophet inequalities in the work of
Krengel, Sucheston, and Garling~\cite{KS77,KS78}
was discussed earlier.  It would be impossible in this 
amount of space to do justice to the extensive literature
on prophet inequalities.  Of particular relevance to our
work are the so-called \emph{multiple-choice prophet
inequalities} in which either the gambler, the prophet,
or both are given the power to choose more than one
element.  While several papers have been written on
this topic, e.g.~\cite{kennedy85,kennedy87,kertz86},
the near-optimal solution of the most natural case,
in which both the gambler and the prophet have $k>1$
choices, was not completed until the work of Alaei~\cite{Alaei},
who gave a factor-$(1-1/\sqrt{k+3})^{-1}$ prophet inequality
for $k$-choice optimal stopping; a nearly-matching
lower bound of $1+\Omega(k^{-1/2})$ was already known
from prior work.

Research on the relationship between algorithmic mechanism
design and prophet inequalities was initiated by 
Hajiaghayi, Kleinberg, and Sandholm~\cite{HKS-aaai},
who observed that algorithms used in the derivation of
prophet inequalities, owing to their monotonicity
properties, could be interpreted as truthful online
auction mechanisms and that the prophet inequality 
in turn could be interpreted as the mechanism's 
approximation guarantee.  Chawla et al.~\cite{CHMS}
discovered a much subtler relation between the two
subjects: questions about the approximability of 
offline Bayesian optimal mechanisms by sequential
posted-price mechanisms could be translated into
questions about prophet inequalities, via the use
of virtual valuation functions.  A fuller discussion of their
contributions appears earlier in this section.  Recent work
by Alaei~\cite{Alaei} deepens still further the connections
between these two research areas,
obtaining a near-optimal $k$-choice prophet inequality
and applying it to a much more general Bayesian combinatorial
auction framework than that studied in~\cite{CHMS}.

While not directly related to our work, the 
matroid secretary problem~\cite{BIK} also concerns
relations between optimal stopping and matroids,
this time under  the assumption of a randomly ordered
input, rather than independent random numbers in a 
fixed order.  In fact, the ``hard examples'' for
many natural examples in the matroid-secretary
setting also translate into hard examples for the
prophet inequality setting.  In light of this 
relation, it is intriguing that our work solves
the matroid prophet inequality problem whereas the
matroid secretary problem remains unsolved, despite
intriguing progress on special cases~\cite{dimitrov,korula}, 
general matroids~\cite{soda12},
and relaxed versions of the problem~\cite{soto}.

Finally, the Bayesian online selection problem that we 
consider here can be formulated as an exponential-sized
Markov decision process, whose state reflects the entire
set of decisions made prior to a specified point during
the algorithm's execution.  Thus, our paper can be 
interpreted as a contribution to the growing CS literature
on approximate solutions of exponential-sized Markov
decision processes, e.g.~\cite{dgv,gm-switching,gm-restless}.
Most of these papers use LP-based techniques.
Combinatorial algorithms based on simple thresholding rules,
such as ours,
are comparatively rare although there are some other examples
in the literature on such problems, for example~\cite{goel}.

\section{Preliminaries}

\paragraph{Bayesian online selection problems.}
An instance of the Bayesian online selection problem (BOSP)
is specified by a ground set $\gset$, a downward-closed
set system $\isets \subseteq 2^{\gset}$, and for each
$x \in \gset$ a probability distribution $\distr_x$ 
supported on the set $\R_+$ of non-negative real numbers.
These data determine a probability distribution over
functions $\wt : \gset \to \R_+$, in which the random
variables $\{ \wt(x) \mid x \in \gset\}$ are independent 
and $\wt(x)$ has distribution $\distr_x$.  We refer to
$\wt(x)$ as the \emph{weight of $x$}, and we extend $\wt$
to an additive set function defined on $2^{\gset}$ by
$\wt(A) = \sum_{x \in A} \wt(x)$.  Elements of 
$\isets$ are called \emph{feasible sets}.  
For a given assignment of weights, $\wt$,
we let $\MAX(\wt)$ denote the maximum-weight
feasible set and $\opt(w)$ denotes its weight;
we will abbreviate these to $\MAX$ and $\opt$
when the weights $\wt$ are clear from context.

An \emph{input sequence} is a sequence $\iseq$ of ordered
pairs $(x_i,\wt_i)\; i=1,\ldots,n$, each belonging to  
$\gset \times \R_+$, such that every element of
$\gset$ occurs exactly once in the sequence $x_1,\ldots,x_n$.
A \emph{deterministic online selection algorithm}
is a function $A$ mapping every input sequence $\iseq$ to a 
set $A(\iseq) \in \isets$ such that for any two input
sequences $\iseq,\iseq'$ that match on the first $i$ pairs
$(x_1,w_1),\ldots,(x_i,w_i)$, the sets
$A_i(\iseq) = A(\iseq) \cap \{1,\ldots,i\}$
and $A_i(\iseq') = A(\iseq') \cap \{1,\ldots,i\}$ are
identical.  A \emph{randomized online
selection algorithm} is a probability distribution over
deterministic ones.  The algorithm's choices
define \emph{decision variables} $\dv_i(\iseq)$ which are
indicator functions of the events $x_i \in A(\iseq)$.
An algorithm is \emph{monotone} if increasing the value 
of $w_i$ (while leaving the rest of $\iseq$ unchanged) 
cannot decrease the value of $\expect{\dv_i(\iseq)}$, where
the expectation is over the algorithm's internal randomness
but not the randomness of $\iseq$ (if any).  A monotone deterministic
online selection algorithm can be completely described by a sequence
of thresholds $T_1(\iseq),\ldots,T_n(\iseq)$, where $T_i(\iseq) 
\in \R_+ \cup \{\infty\}$ is
the infimum of the set of weights $w$ 
such that $i \in A(\iseq')$ when $\iseq'$
is obtained from $\iseq$ by modifying $w_i$ to $w$.
Conversely, for any sequence of threshold functions
$T_1,\ldots,T_n$ such that $T_i(\iseq)$ depends
only on the first $i-1$ elements of $\iseq$ and
$T_i(\iseq) = \infty$ whenever $A_{i-1}(\iseq) \cup \{i\} 
\not\in \isets$, there is a corresponding
monotone deterministic online selection
algorithm that selects $x_i$ whenever $w_i \geq T_i(\iseq)$.

Notice that an algorithm as defined above is agnostic to the 
order in which the matroid elements will be presented, i.e.\ it
has a well-defined behavior no matter what order the elements
appear in the input sequence.  One could also consider
\emph{order-aware} algorithms that know the entire 
sequence $x_1,\ldots,x_n$ in advance 
(but not the weights $w_1,\ldots,w_n$).
In the matroid setting, our factor-2 prophet inequality for order-agnostic
algorithms reveals that order-aware algorithms have no advantage over
order-agnostic ones in the worst case; it is an interesting open
question whether the same lack of advantage holds more generally.


One can similarly distinguish between adversaries with respect
to their power to choose the ordering of the sequence. The original BOSP treated in previous work \cite{KS77,KS78} considers a \emph{fixed-order} adversary. That is, the adversary chooses an ordering (or distribution over orderings) for revealing the elements of $\gset$ \emph{without knowing any of the weights $w(x)$}. Our main result is an algorithm that achieves $\frac{1}{2}\opt$ (or $\frac{1}{4p-2}\opt$) against a fixed-order adversary. This result combined with the techniques of \cite{CHMS} immediately yields OPMs for single-parameter mechanism design. 
To extend our results to BMUMD, we must consider a stronger type of adversary.  There are many ways that an adversary could adapt to the sampled weights and/or the algorithm's decisions, some more powerful than others.  The type of adaptivity that is relevant to our paper will be called an \emph{online weight-adaptive} adversary. An online weight-adaptive adversary chooses the next element of $\gset$ to reveal one at a time. After choosing $x_1,\ldots,x_{i-1}$ and learning $w(x_1),\ldots,w(x_{i-1})$, the online weight-adaptive adversary chooses the next $x_i$ to reveal \emph{without knowing the weight $w(x_i)$} (or any weights besides $w(x_1),\ldots,w(x_{i-1})$). Fortunately, the same exact proof shows that our algorithm, without any modification, also achieves $\frac{1}{2}\opt$ (or $\frac{1}{4p-2}\opt$) against an online weight-adaptive adversary. The connection between BMUMD and online weight-adaptive adversaries is not trivial, and is explained in Section \ref{sec:opm}.

\paragraph{Matroids.}
A \emph{matroid} $\mtd$ consists of a ground set
$\gset$ and a nonempty downward-closed set system $\isets \subseteq
2^{\gset}$ satisfying the \emph{matroid exchange axiom}:
for all pairs of sets $I,J \in \isets$ such that 
$|I| < |J|$, there exists an element $x \in J$ such
that $I \cup \{x\} \in \isets$.  Elements of $\isets$
are called \emph{independent sets} when $(\gset,\isets)$
is a matroid.  A maximal independent set is called a
\emph{basis}.  If $A$ is a subset of $\gset$, its 
\emph{rank}, denoted by 
$\matrank(A)$, is the maximum cardinality
of an independent subset of $A$.  Its \emph{closure}
or \emph{span}, denoted by $\matspan(A)$, is the set
of all $x \in \gset$ such that $\matrank(A \cup \{x\}) = 
\matrank(A)$.  It is well known that the following greedy
algorithm selects a maximum-weight basis of a matroid:
number the elements of $\gset$ as $x_1,\ldots,x_n$ in 
decreasing order of weight, and select the set of all
$x_i$ such that $x_i \not\in \matspan(\{x_1,\ldots,x_{i-1}\})$.

\section{Algorithms for Matroids}
\label{sec:matroid}

In this section we prove our main theorem, asserting
the existence of algorithms whose expected reward
is at least $\frac12 \opt$ when playing against any online weight-adaptive adversary. Here is some intuition as to the considerations guiding the design of our algorithm. Imagine a prophet that is forced to start by accepting the set $A$, and let the \emph{remainder} of $A$ (denoted $R(A)$, defined formally in the following section) denote the subset that the restricted prophet adds to $A$. Let the \emph{cost} of $A$ (denoted $C(A)$, defined formally in the following section) denote the subset that the unrestricted prophet selected in place of $A$. Then the restricted prophet makes $\wt(A) + \expect{\wt(R(A))}$ in expectation, while the unrestricted prophet makes $\expect{\wt(C(A))} + \expect{\wt(R(A))}$. So if $A$ satisfies $\wt(A) \geq \frac{1}{\alpha} \expect{\wt(C(A))}$ for a small constant $\alpha$, it is not so bad to get stuck holding set $A$.  However, just because $A$ is not a bad set to start with does not mean we shouldn't accept anything that comes later.  After all, the empty set is not a bad set to start with. If we can choose $A$ in a way such that for any $V$ we reject with $A \cup V \in \isets$, $\wt(V) \leq \frac{1}{\alpha}\expect{\wt(R(A))}$, then $A$ is not a bad set to finish with. Simply put, we want to choose thresholds that are large enough to guarantee that $\wt(A)$ compares well to $\expect{\wt(C(A))}$, but small enough to guarantee that everything we reject is not too heavy. Indeed, the first step in our analysis is to define this property formally and show that an algorithm with this property obtains a $\frac{1}{\alpha}$-approximation.

\subsection{Detour: The rank-one case}
\label{sec:rank-one}

To introduce the ideas underlying our algorithm and its
analysis, we start with a very simple analysis of the
case of rank-one matroids.  This is the special case
of the problem in which the algorithm is only allowed to 
make one selection, i.e.~the same setting as the original
prophet inequality~\eqref{eq:prophet-inequality}.  Thus,
the algorithm given in this section can be regarded as
providing a new and simple proof of that inequality.

Let the random weights of the elements by denoted by
$X_1,\ldots,X_n$, and let $T = \expect{\max_i X_i}/2$. 
We will show that an algorithm that stops at the first
time $\tau$ such that $X_\tau \geq T$
makes at least $T$ in expectation. 
Let $p = \Pr[\max_i X_i \geq T]$. 
Then we get the following inequality, for any $x > T$:
$$\Pr[X_\tau > x] \geq (1-p) \sum_{i=1}^n \Pr[X_i > x]$$
This is true because with probability $1-p$ the algorithm accepts nothing, so with probability at least $(1-p)$ it has accepted nothing by the time it processes $X_i$. So the probability that the algorithm accepts $X_i$ and that $X_i > x$ is at least $(1-p)\Pr[X_i > x]$. It is also clear, by the union bound, that
$$\sum_{i=1}^n \Pr[X_i > x] \geq \Pr[\max_i X_i > x]$$
and therefore, for all $x > T$, 
$$\Pr[X_\tau > x] \geq (1-p)\Pr[\max_i X_i > x].$$
Now, observe that $\expect{\max_i X_i} = \int^T_0 \Pr[\max_i X_i > x] \, dx 
+ \int^\infty_T \Pr[\max_i X_i > x] \, dx = 2T$. 
As the first term is clearly at most $T$, 
the second term must be at least $T$. So finally, we write:
\begin{align*}
\mathbb{E}[X_\tau] &= \int_0^T \Pr[X_\tau > x] \, dx + 
     \int^\infty_T \Pr[X_\tau > x]\, dx \\
& \geq pT + (1-p) \int^\infty_T \Pr[\max_i X_i > x] \, dx \\
& \geq pT + (1-p)T = T = \frac12 \expect{\max_i X_i}
\end{align*}
which completes the proof of~\eqref{eq:prophet-inequality}.

\subsection{A property guaranteeing $\alpha$-approximation}

To design and analyze algorithms for general matroids, we begin by 
defining a property of a deterministic monotone algorithm that
we refer to as \emph{$\alpha$-balanced thresholds}.
In this section we prove that the expected reward of any such algorithm 
is at least $\frac{1}{\alpha} \opt$.  In the following section we 
construct an algorithm with 2-balanced thresholds,
completing the proof of the main theorem.

To define $\alpha$-balanced thresholds, we must first
define some notation.  Let $\wt, \wt' : \gset \to \R_+$ 
denote two assignments of weights to $\gset$, 
both sampled indepedently from the given distribution.
We consider running the algorithm on an input sequence
$\iseq = (x_1,\wt(x_1)), \ldots, (x_n, \wt(x_n))$
and comparing the value of its selected set,
$A=A(\iseq)$, with that of the basis $B$ that maximizes
$\wt'(B)$.  The matroid exchange axiom ensures
that there is at least one way to
partition $B$ into disjoint subsets $C,R$
such that $A \cup R$ is also a basis of $\mtd$.
(Consider adding elements of $B$ one-by-one to $A$,
preserving membership in $\isets$, until the two sets
have equal cardinality, and let $R$ be the set of 
elements added to $A$.)
Among all such partitions, let $C(A),R(A)$ denote
the one that maximizes $\wt'(R)$.

\begin{definition} \label{def:alphabeta}
For a parameter $\alpha>0$, a deterministic
monotone algorithm has \emph{$\alpha$-balanced thresholds}
if it has the following
property.  For every input sequence $\iseq$, if 
$A = A(\iseq)$ and $V$ is a set disjoint from $A$ 
such that $A \cup V \in \isets$, then
\begin{align}
\label{eq:alpha}
\sum_{x_i \in A} T_i(\iseq) & \geq 
\left( \frac{1}{\alpha} \right) \cdot \expect{\wt'(C(A))} \\
\label{eq:beta}
\sum_{x_i \in V} T_i(\iseq) & \leq 
\left(1-\frac{1}{\alpha} \right) \cdot \expect{\wt'(R(A))},
\end{align}
where the expectation is over the random choice of $\wt'$.
\end{definition}

\begin{proposition} \label{prop:alphabeta}
If a monotone algorithm has $\alpha$-balanced thresholds,
then it satisfies the following approximation
guarantee against online weight-adaptive adversaries:
\begin{equation} \label{eq:apx}
\expect{\wt(A)} \geq 
\frac{1}{\alpha} \opt.
\end{equation}
\end{proposition}
\begin{proof}
We have
\begin{equation} \label{eq:opt}
\opt = \expect{\wt'(C(A)) + \wt'(R(A))}
\end{equation}
because $C(A) \cup R(A)$ is a maximum-weight basis with respect
to $\wt'$, and $\wt'$ has the same distribution as $\wt$.
For any real number $z$, we will use the 
notation $(z)^+$ to denote $\max \{z,0\}$.  
The proof will consist of deriving the following three inequalities,
in which $w_i$ stands for $w(x_i)$.
\begin{align} \label{eq:step1}
\expect{\sum_{x_i \in A} T_i} & \geq \frac{1}{\alpha} \expect{\wt'(C(A))} \\ 
\label{eq:step2}
\expect{\sum_{x_i \in A} (\wt_i - T_i)^+} & \geq 
\expect{\sum_{x_i \in R(A)} (\wt'(x_i) - T_i)^+} \\
\label{eq:step3}
\expect{\sum_{x_i \in R(A)} (\wt'(x_i) - T_i)^+} & \geq
\frac{1}{\alpha} \expect{\wt'(R(A))}.
\end{align}
Summing~\eqref{eq:step1}-\eqref{eq:step3} and using the 
fact that $T_i + (\wt_i-T_i)^+ = \wt_i$ for all $x_i \in A$, we obtain
\begin{equation*}
\expect{\wt(A)} \geq \frac{1}{\alpha}
\expect{\wt'(C(A))} + \frac{1}{\alpha}
\expect{\wt'(R(A))}.
\end{equation*}

Inequality~\eqref{eq:step1} 
is a restatement of
the definition
of $\alpha$-balanced thresholds.
Inequality~\eqref{eq:step2} is deduced from the following observations.
First, the algorithm selects every $i$ such that $\wt_i > T_i$, so
$\sum_{x_i \in A} (\wt_i - T_i)^+ = \sum_{i=1}^n (\wt_i - T_i)^+$.
Second, the online property of the algorithm and the fact that weight-adaptive adversaries do not learn $\wt_i$ before choosing to reveal $x_i$ imply that $T_i$
depends only on $(x_1,\wt_1),\ldots,(x_{i-1},\wt_{i-1})$ and that the
random variables $\wt(x_i), \wt'(x_i), T_i$ are independent.
As $\wt_i = \wt(x_i)$ and $\wt'(x_i)$ are identically distributed, it
follows that 
\[
\expect{\sum_{i=1}^n (\wt_i-T_i)^+} =
\expect{\sum_{i=1}^n (\wt'(x_i) - T_i)^+} \geq
\expect{\sum_{x_i \in R(A)} (\wt'(x_i) - T_i)^+},
\]
and~\eqref{eq:step2} is established.  Finally, we
apply Property~\eqref{eq:beta} of $\alpha$-balanced
thresholds,
using the set $V = R(A)$, to deduce that
\begin{align*}
\expect{\sum_{x_i \in R(A)} \wt'(x_i)} & \leq
\expect{\sum_{x_i \in R(A)} T_i} + 
\expect{\sum_{x_i \in R(A)} (\wt'(x_i) - T_i)^+}  \\
& \leq
\left( 1 - \frac{1}{\alpha} \right) \, \expect{\sum_{x_i \in R(A)} \wt'(x_i)}
+
\expect{\sum_{x_i \in R(A)} (\wt'(x_i) - T_i)^+}  \\
\frac{1}{\alpha} \expect{\sum_{x_i \in R(A)} \wt'(x_i)} & \leq
\expect{\sum_{x_i \in R(A)} (\wt'(x_i) - T_i)^+} 
\end{align*}
Consequently~\eqref{eq:step3} holds, which concludes the proof.
\end{proof}
\subsection{Achieving 2-balanced thresholds}
\label{sec:2-balanced}

This section presents an algorithm with 2-balanced
thresholds.  The algorithm is quite simple.  In step
$i$, having already selected the (possibly empty)
set $A_{i-1}$, we set threshold $T_i = \infty$ if 
$A_{i-1} \cup \{x_i\} \not\in \isets$, and otherwise
\begin{align} \label{eq:ti-r}
T_i &= \tfrac12 \expect{\wt'(R(A_{i-1})) - \wt'(R(A_{i-1} \cup \{x_i\}))}\\
&= \tfrac12 \expect{\wt'(C(A_{i-1} \cup \{x_i\})) - 
\wt'(C(A_{i-1}))}
\label{eq:ti-c}
\end{align}
The algorithm selects element $x_i$ if and only if $w_i \geq T_i$.
The fact that both~\eqref{eq:ti-r} and~\eqref{eq:ti-c} define the
same value of $T_i$ is easy to verify.  Let $B$ denote the
maximum weight basis of $\mtd$ with weights $\wt'$.
\begin{align*}
\wt'(C(A_{i-1})) + \wt'(R(A_{i-1})) =
\wt'(B) & =
\wt'(C(A_{i-1} \cup \{x_i\})) +
\wt'(R(A_{i-1} \cup \{x_i\})) \\
\wt'(R(A_{i-1})) - \wt'(R(A_{i-1} \cup \{x_i\})) &= 
\wt'(C(A_{i-1} \cup \{x_i\})) - \wt'(C(A_{i-1})) 
\end{align*}

Property~\eqref{eq:alpha} in the definition of $\alpha$-balanced
thresholds follows from a telescoping sum.
\begin{align*}
\sum_{x_i \in A} T_i & = 
\tfrac12 \sum_{x_i \in A} \expect{\wt'(C(A_{i-1} \cup \{x_i\})) - 
\wt'(C(A_{i-1}))} \\
& = 
\tfrac12 \sum_{x_i \in A} \expect{\wt'(C(A_i)) - \wt'(C(A_{i-1}))} \\
& = 
\tfrac12 \expect{\wt'(C(A_n)) - \wt'(C(A_0))} =
\tfrac12 \expect{\wt'(C(A))}.
\end{align*}

The remainder of this section is devoted to proving
Property~\eqref{eq:beta} in the definition of $\alpha$-balanced
thresholds.  In the present context, with $\alpha=2$ and thresholds $T_i$ 
defined by~\eqref{eq:ti-r}, the property simply asserts that
for every pair of disjoint sets $A,V$ such that 
$A \cup V \in \isets$,
\begin{align*}
\expect{\sum_{x_i \in V} \wt'(R(A_{i-1})) - \wt'(R(A_{i-1} \cup \{x_i\}))} 
 = 2 \sum_{x_i \in V} T_i(\iseq)  & \leq 
\expect{\wt'(R(A))} 
\end{align*}
We will show, in fact, that this inequality holds 
for every non-negative weight assignment $\wt'$ and not merely in expectation.
The proof appears in Proposition~\ref{prop:done} below.  
To establish it, we will need some basic properties of matroids.
\begin{definition}[\cite{schrijver-b}, Section 39.3] 
\label{def:del-contr-trunc}
If $\mtd$ is a matroid and $S$ is a subset
of its ground set, 
the \emph{deletion} $\mtd - S$ and the 
\emph{contraction} $\mtd / S$ are two matroids with ground set
$\gset - S$.  A set $T$ is independent in $\mtd - S$ 
if $T$ is independent in $\mtd$, whereas $T$ is independent in
$\mtd/S$ if $T \cup S_0$ is independent in $\mtd$, where 
$S_0$ is any maximal independent subset of $S$.
\end{definition}
\begin{lemma} \label{lem:matching}
Suppose $\mtd = (\gset,\isets)$ is a matroid 
and $V,R \in \isets$ are two independent sets of equal cardinality.
\begin{enumerate*}
\item \label{unweighted}
There is a bijection
$\phi: V \to R$ such that for every $v \in V$,
$(R - \{\phi(v)\}) \cup \{v\}$
is an independent set.
\item \label{weighted}
For a weight function $\wt' : \gset \to \R$, 
suppose that $R$ has the maximum weight of all 
$|R|$-element independent subsets of $V \cup R$.
Then the bijection $\phi$ in part~\ref{unweighted}
also satisfies $\wt'(\phi(v)) \geq \wt'(v)$.
\end{enumerate*}
\end{lemma}
\begin{proof}
Part~\ref{unweighted} is Corollary 39.12a in~\cite{schrijver-b}.
To prove part~\ref{weighted}, simply observe that
the weight of $(R - \{\phi(v)\}) \cup \{v\}$ cannot
be greater than the weight of $R$, by our assumptions on $R$
and $\phi$.
\end{proof}

The next two lemmas establish basic properties of 
the function $S \mapsto R(S)$.

\begin{lemma} \label{lem:greedy-r}
For any independent set $A$, the set $R(A)$ is equal to the 
maximum weight basis of $\mtd/A$.
\end{lemma}
\begin{proof}
Let $B$ be the maximum-weight basis of $\mtd$.
Among all bases of $\mtd/A$ that are contained in $B$,
the set $R(A)$ is, by definition, the one of maximum
weight.  Therefore, if it is not the maximum-weight
basis of $\mtd/A$, the only reason can be that there
is another basis of $\mtd/A$, not contained in $B$,
having strictly greater weight.  But we know that the
maximum-weight basis of $\mtd/A$ is selected by the
greedy algorithm, which iterates through the list
$y_1,\ldots,y_k$ of elements of $\gset - A$ sorted
in order of decreasing weight, and picks each element
$y_i$ that is not contained in $\matspan(A \cup \{y_1,\ldots,y_{i-1}\})$.
In particular, every $y_i$ chosen by the greedy algorithm
on $\mtd/A$ satisfies $y_i \not\in \matspan(\{y_1,\ldots,y_{i-1}\})$
and therefore belongs to $B$.  Thus the maximum-weight
basis of $\mtd/A$ is contained in $B$ and must equal $R(A)$.
\end{proof}
\begin{lemma} \label{lem:submod}
For any independent set $J$, the function $f(S) = \wt'(R(S))$
is a submodular set function on subsets of $J$.
\end{lemma}
\begin{proof}
For notational convenience, in this proof we will denote the
union of two sets by `$+$' rather than `$\cup$'.  Also, we will
not distinguish between an element $x$ and the singleton set $\{x\}$.

To prove
submodularity it suffices to consider an independent set
$S+x+y$ and to prove that $f(S)-f(S+x) \leq f(S+y)-f(S+x+y)$.
Replacing $\mtd$ by $\mtd/S$, we can reduce to the case that
$S = \emptyset$ and prove that $f(\emptyset)-f(x) \leq
f(y) - f(x+y)$ whenever $\{x,y\}$ is a two-element independent
set.

What is the interpretation of $f(\emptyset)-f(x)$?  Recall
that $f(\emptyset)=\wt'(R(\emptyset))$ is the weight of the
maximum-weight basis $B$ of $\mtd$.  Similarly, $f(x)$ is the 
weight of the maximum-weight basis $B_x$ of $\mtd/\{x\}$.  
Let $b_1,b_2,\ldots,b_r$ denote the elements of $B$ in 
decreasing order of weight.  Consider running two 
executions of  the greedy algorithm to select $B$ and $B_x$
in parallel.  The only step in which the algorithms make
differing decisions is the first step $i$ in which 
$\{b_1,\ldots,b_i\} \cup \{x\}$ contains a circuit.
In this step, $b_i$ is included in $B$ but excluded from $B_x$.
Similarly, when we run two executions of the greedy algorithm
to select $B_y$ and $B_{xy}$ --- the maximum-weight bases of
$\mtd/\{y\}$ and $\mtd/\{x,y\}$, respectively --- the
only step in which differing decisions are made is the earliest
step $j$ in which $\{b_1,\ldots,b_j\} \cup \{x,y\}$ contains a 
circuit.  But $j$ certainly cannot be later than $i$, since
$\{b_1,\ldots,b_i\} \cup \{x,y\}$ is a superset of
$\{b_1,\ldots,b_i\} \cup \{x\}$ and hence contains a circuit.
We may conclude that
\[
f(\emptyset) - f(x) = b_i \leq b_j = f(y) - f(x+y),
\]
and hence $f$ is submodular as claimed.
\end{proof}
\begin{proposition} \label{prop:done}
For any disjoint sets $A,V$ such that $A \cup V \in \isets$,
$$
\sum_{x_i \in V} \wt'(R(A_{i-1})) - \wt'(R(A_{i-1} \cup \{x_i\}))
\leq
\wt'(R(A)).
$$
\end{proposition}
\begin{proof}
The function $f(S) = \wt'(R(S))$ is submodular
on subsets $S \subseteq A \cup V$, by Lemma~\ref{lem:submod}.
Hence 
\begin{equation} \label{eq:done-1}
\sum_{x_i \in V} \wt'(R(A_{i-1})) - \wt'(R(A_{i-1} \cup \{x_i\}))
\leq
\sum_{x \in V} \wt'(R(A)) - \wt'(R(A \cup \{x\})).
\end{equation}
Apply Lemma~\ref{lem:matching} to the independent sets
$V, R(A)$ in $\mtd/A$ to obtain a bijection $\phi$ such 
that $\wt'(\phi(x)) \geq \wt'(x)$ 
and $A \cup (R(A) - \phi(x)) \cup \{x\} \in \isets$ for all $x \in V$.
By definition of $R(\cdot)$, we know that
$A \cup \{x\} \cup R(A \cup \{x\})$ is the maximum weight 
independent subset of $A \cup \{x\} \cup B$ that 
contains $A \cup \{x\}$.  One such set is
$A \cup (R(A) - \phi(x)) \cup \{x\}$, so
\begin{align}
\nonumber
\wt'(A) + \wt'(R(A)) - \wt'(\phi(x)) + \wt'(x) &\leq
\wt'(A) + \wt'(R(A \cup \{x\})) + \wt'(x) \\
\nonumber
\wt'(R(A)) - \wt'(R(A \cup \{x\})) & \leq
\wt'(\phi(x)) \\
\sum_{x \in V} \wt'(R(A)) - \wt'(R(A \cup \{x\})) & \leq
\sum_{x \in V} \wt'(\phi(x)) = \wt'(R).
\label{eq:done-2}
\end{align}
The proposition follows by combining~\eqref{eq:done-1}
and~\eqref{eq:done-2}.
\end{proof}

\section{Matroid intersections}
\label{sec:mat-int}

Our algorithm and proof for matroid intersections is quite similar. We need to modify some definitions and extend some proofs, but the spirit is the same.

\subsection{A generalization of $\alpha$-balanced thresholds}
We first have to extend our notation a bit. Denote the independent sets for the $p$ matroids as $\isets_1,\ldots,\isets_p$. Denote the ``truly independent'' sets as $\isets = \cap_j \isets_j$. Still let $\wt, \wt' : \gset \to \R_+$ 
denote two assignments of weights to $\gset$, 
both sampled indepedently from the given distribution.
We consider running the algorithm on an input sequence
$\iseq = (x_1,\wt(x_1)), \ldots, (x_n, \wt(x_n))$
and comparing the value of its selected set,
$A=A(\iseq)$, with that of the $B \in \isets$ that maximizes
$\wt'(B)$.  For all $j$, the matroid exchange axiom ensures
that there is at least one way to
partition $B$ into disjoint subsets $C_j,R_j$
such that $A \cup R_j \in \isets_j$, and $B \subseteq \matspan_j(A \cup R_j)$.
Among all such partitions, let $C_j(A),R_j(A)$ denote
the one that maximizes $\wt'(R_j)$ (greedily add elements from $B$ to $R_j$ unless it creates a dependency in $\isets_j$). We denote by $R(A) = \cap_j R_j(A)$ and $C(A) = \cup_j C_j(A)$. 

\begin{definition} \label{def:alphaintersection}
For a parameter $\alpha>0$, a deterministic
monotone algorithm has \emph{$\alpha$-balanced thresholds}
if it has the following
property.  For every input sequence $\iseq$, if 
$A = A(\iseq)$ and $V$ is a set disjoint from $A$ 
such that $A \cup V \in \isets$, then
\begin{align}
\label{eq:alphaintersection}
\sum_{x_i \in A} T_i(\iseq) & \geq 
\left( \frac{1}{\alpha} \right) \cdot \expect{\sum_j\wt'(C_j(A))} \\
\label{eq:betaintersection}
\sum_{x_i \in V} T_i(\iseq) & \leq 
\left( \frac{1}{\alpha} \right) \cdot \expect{\sum_j \wt'(R_j(A))},
\end{align}
where the expectation is over the random choice of $\wt'$.
\end{definition}

\begin{proposition} \label{prop:alphabetaintersection}
If a monotone algorithm has $\alpha$-balanced thresholds for $\alpha \geq 2$,
then it satisfies the following approximation
guarantee against weight-adaptive adversaries when $\isets$ is the intersection of $p$ matroids:
\begin{equation} \label{eq:apxintersection}
\expect{\wt(A)} \geq 
\frac{\alpha-p}{\alpha(\alpha-1)} \opt.
\end{equation}
\end{proposition}

The proof closely parallels the proof of Proposition~\ref{prop:alphabeta},
and is given in the appendix.

\subsection{Obtaining $\alpha$-balanced thresholds}
This section presents an algorithm obtaining $\alpha$-balanced thresholds for any $\alpha > 1$. One can take a derivative to see that the optimal choice of $\alpha$ for the intersection of $p$ matroids is $\alpha_p = p + \sqrt{p(p-1)}$. For simplicity, we will instead just use $\alpha = 2p$, as this is nearly optimal and always at least $2$. When $\alpha = 2p$, the approximation guarantee from Proposition \ref{prop:alphabetaintersection} is $\frac{1}{4p-2}$.

We now define our thresholds. Let 
\begin{align*}
T(A,i,j) &= \frac{1}{\alpha}\expect{w'(R_j(A)) - w'(R_j(A\cup\{x_i\}))} \\
&= \frac{1}{\alpha} \expect{w'(C_j(A \cup \{x_i\})) - w'(C_j(A))} \\
T(A,i) &= \sum_j T(A,i,j).
\end{align*}
In step $i$, having already selected the (possibly empty) set $A_{i-1}$, we set threshold $T_i = \infty$ if $A_{i-1} \cup \{i\} \notin \isets$, and $T_i = T(A_{i-1},i)$ otherwise. In other words, each $T(A,i,j)$ is basically the same as the threshold used for the single matroid algorithm if $\isets_j$ was the only matroid constraint. It is not exactly the same, because $R(A)$ when $\isets_j$ is the only matroid is not the same as $R_j(A)$ in the presence of other matroid constraints. $T(A,i)$ just sums $T(A,i,j)$ over all matroids.\\
\\
The proof of Equation \eqref{eq:alphaintersection} follows exactly the proof of Equation \eqref{eq:alpha}.\\
\\
The proof of Equation \eqref{eq:betaintersection} follows from Proposition \ref{prop:done}, although perhaps not obviously. As $A \cup V \in \isets$, we clearly have $A \cup V \in \isets_j$ for all $j$. So the hypotheses of Proposition \ref{prop:done} are satisfied for all $j$. Summing the bound we get in Proposition \ref{prop:done} over all $j$ gives us Equation \eqref{eq:betaintersection}.

\section{Lower Bounds} \label{sec:lb}
Here we provide two examples. The first is the well-known example of~\cite{KS77} showing that the factor of $2$ is tight for matroids.  We present their construction here for completeness. The second shows that the ratio $O(p)$ is tight for the intersection of $p$ matroids.

We start with the well-known example of \cite{KS77}. Consider the $1$-uniform matroid over $2$ elements. We have $\wt(1) = 1$ with probability $1$, $\wt(2) = n$ with probability $1/n$ and $0$ otherwise. Then the prophet obtains $2-1/n$ in expectation, but the gambler obtains at most $1$, as his optimal strategy is just to take the first element always.

The example for the intersection of $p$ matroids has appeared in other forms in \cite{BIK,CHMS}. Let $q$ be a prime between $p/2$ and $p$. Then let $\gset = \{(i,j) \,|\, 0 \leq i \leq q^q-1, 0 \leq j \leq q-1\}$. Then let $\isets$ contain all sets of the form $\{(i,j_1),\ldots,(i,j_x)\}$. Now let $w(i,j) = 1$ with probability $1/q$, and $w(i,j) = 0$ otherwise, for all $i,j$. Reveal the elements in any order. No matter what strategy the gambler uses to pick the first element, his optimal strategy from that point on is to just accept every remaining element with the same first coordinate. However the gambler winds up with his first element, he makes at most $1-1/q$ in expectation from the remaining elements he is allowed to pick (as there are at most $q-1$ remaining elements, and each has $\expect{w(i,j)} = 1/q$). Therefore, the expected payoff to the gambler is less than $2$. However, with probability at least $(1-1/e)$, there exists an $i$ such that $w(i,j) = 1$ for all $j$ (as the probability that this occurs for a fixed $i$ is $1/q^q$ and there are $q^q$ different $i$'s). So the expected payoff to the prophet is $\Theta(q)$. 

Finally, we just have to show that $\isets$ can be written as the intersection of $q$ matroids. Let $\isets_x$ be the partition matroid that partitions $\gset$ into $\sqcup_j S_j = \sqcup_j \cup_i \{(i,xi+j \pmod q)\}$, and requires that only one element of each $S_j$ be chosen. Then clearly, $\isets \subseteq \cap_{x \in \mathbb{Z}_q} \isets_x$ as any two elements with the same first coordinate lie in different partitions in each of the $\isets_x$. In addition, $\cap_{x \in \mathbb{Z}_q} \isets_x \subseteq \isets$. Consider any $(i,j)$ and $(i',j')$ with $i \neq i'$. Then when $(j-j') \pmod q = x(i-i') \pmod q$, $(i,j)$ and $(i',j')$ are in the same partition of $\isets_x$. As $q$ is prime, this equation always has a solution. Therefore, we have shown that $\isets = \cap_{x \in \mathbb{Z}_q} \isets_x$, and $\isets$ can be written as the intersection of $q\leq p$ matroids. As the prophet obtains $\Theta(p)$ in expectation, and the gambler obtains less than $2$ in expectation, 
no algorithm can achieve an approximation factor better than $O(p)$.

\section{Interpretation as OPMs} \label{sec:opm}
Here, we describe how to use our algorithm to design OPMs for unit-demand multi-parameter bidders under matroid and matroid intersection feasibility constraints.  We begin by recalling the definition of Bayesian multi-parameter unit-demand mechanism design (BMUMD) from~\cite{CHMS}.  In any such mechanism design problem, there is a set of services, $\gset$, partitioned into disjoint subsets $J_1,\ldots,J_n$, one for each bidder.  The mechanism must allocate a set of services, subject to downward-closed feasibility constraints given by a collection $\isets$ of feasible subsets.  We assume that the feasibility constraints guarantee that no bidder receives more than a single service, i.e.\ that the intersection of any feasible set with one of the sets $J_i$ contains no more than one element.  (If this property is not already implied by the given feasibility constraints, it can be ensured by intersecting the given constraints with one additional partition matroid constraint.)  

As in the work of Chawla et al.~\cite{CHMS}, 
we assume that each bidder $i$'s values for the 
services in set $J_i$ are independent random variables,
and we analyze BMUMD mechanisms for any such distribution
by exploring a closely-related 
single-parameter domain that we denote by $\icopies$.
In $\icopies$ there are $|\gset|$ bidders, each
of whom wants just a single service $x$ and has a value $v_x$
for receiving that service.  The feasibility constraints
are the same in both domains --- the mechanism may select any set of services 
that belongs to $\isets$ --- and the joint distribution of
the values $v_x \; (x \in \isets)$ is the same as well; the only
difference between the two domains is that an individual bidder $i$
in the BMUMD problem becomes a set of competing bidders
(corresponding to the elements of $J_i$) in the domain
$\icopies$.  As might be expected, the increase in competition
between bidders results in an increase in revenue for the
optimal mechanism; indeed, the following lemma from~\cite{CHMS}
will be a key step in our analysis.

\begin{lemma} \label{lem:chms}
Let $\cal A$ be any individually rational and truthful deterministic mechanism for instance $\isets$ of BMUMD. Then the expected revenue of $\cal A$ is no more than the expected revenue of the optimal mechanism for $\icopies$.
\end{lemma}

A second technique that we will borrow from~\cite{CHMS} (and, ultimately,
from Myerson's original paper on optimal mechanism design~\cite{Mye81}), is 
the technique of analyzing the expected revenue of mechanisms indirectly
via their virtual surplus.  We begin by reviewing the definitions of
virtual valuations and virtual surplus.  Assume that $v_x$, the value of 
bidder $i$ for item $x \in J_i$, has cumulative distribution function
$F_x$ whose density $f_x$ is well-defined and positive on the interval
on which $v_x$ is supported.  Then the virtual valuation function 
$\phi_x$ is defined by
\[
\phi_x(v) = v - \frac{1 - F_x(v)}{f_x(v)},
\]
and the \emph{virtual surplus} of an
allocation $A \in \isets$ is defined to be the sum
$\sum_{x \in A} \phi_x(v_x)$.
Myerson~\cite{Mye81} proved the following:

\begin{lemma} \label{lem:mye}
In single-parameter domains whose bidders have independent
valuations with monotone increasing virtual valuation
functions, the expected revenue of any mechanism in
Bayes-Nash equilibrium is equal to its expected virtual
surplus.
\end{lemma}

The distribution of $v_x$ is said to be \emph{regular}
when the virtual valuation function $\phi_x$ is monotonically increasing.
We will assume throughout the rest of this section that
bidders' values have regular distributions, in order to
apply Lemma~\ref{lem:mye}.  To deal with non-regular 
distributions, it is necessary to use a 
technique known as \emph{ironing},
also due to Myerson~\cite{Mye81}, which
in our context translates into randomized
pricing via a recipe described in Lemma~2 of~\cite{CHMS}.

Our plan is now to design
truthful mechanisms $\mech$ and $\mcopies$ 
for the BMUMD domain $\isets$
and the associated single-parameter domain $\icopies$,
respectively, and to relate them to the optimal mechanisms
for those domains
via the following chain of inequalities.
\begin{equation} \label{eq:opm-outline}
R(\mech) \geq R(\mcopies) = \Phi(\mcopies) \geq
\frac{1}{\alpha} \Phi(\optcopies) =
\frac{1}{\alpha} R(\optcopies) \geq
\frac{1}{\alpha} R(\opt).
\end{equation}
Here, $R(\cdot)$ and $\Phi(\cdot)$ denote the 
expected revenue and expected virtual surplus of
a mechanism, respectively, and
$\alpha$ denotes the approximation guarantee of
a prophet inequality algorithm embedded in our mechanism.
Thus, $\alpha=2$ when $\isets$ is a matroid, and more
generally $\alpha = 4p-2$ when $\isets$ is given by an
intersection of $p$ matroid constraints.

Most of the steps in line~\eqref{eq:opm-outline} are
already justified by the lemmas from prior work
discussed above.  The relation $R = \Phi$ for mechanisms
$\mcopies$ and $\optcopies$ is a consequence of 
Lemma~\ref{lem:mye}, while the relation 
$R(\optcopies) \geq R(\opt)$ is Lemma~\ref{lem:chms}.
We will naturally derive the relation 
$\Phi(\mcopies) \geq \frac{1}{\alpha} \Phi(\optcopies)$
as a consequence of the prophet inequality.  To do so, it suffices
to define mechanism $\mcopies$ such that its allocation
decisions result from running the prophet inequality 
algorithm on an input sequence consisting of the virtual
valuations $\phi_x(v_x)$, presented in an order 
determined by an online weight-adaptive adversary.  
The crux of our proof will consist of designing said
adversary to ensure that the relation 
$R(\mech) \geq R(\mcopies)$ also holds.

Given these preliminaries, we now describe the 
mechanisms $\mech$ and $\mcopies$.  Central to 
both mechanisms is a pricing scheme using thresholds 
$T(A,x)$, defined as the threshold 
$T_s$ that our online algorithm would use at step $s$ 
when $x_s = x$ and the algorithm has accepted the set 
$A$ so far.  
(Contrary to previous sections of
the paper in which steps of the online algorithm's execution
were denoted by $i$, here we reserve the variable $i$ to 
refer to bidders in the mechanism, using $s$ instead to 
denote a step of the online algorithm.
Note that the thresholds assigned by our algorithm
depend only on $A$ and $x$, not on $s$, hence the notation 
$T(A,x)$ is justified.)  Mechanism $\mech$, described by
the pseudocode in Algorithm~\ref{alg:OPM}, simply makes
posted-price offers to bidders $1,2,\ldots,n$ in that order,
defining the posted price for each item by applying its
inverse-virtual-valuation function to the threshold
that the prophet inequality algorithm sets for that item.

\begin{algorithm}[tb]
\caption{Mechanism $\mech$ for unit-demand multi-dimensional bidders}
\begin{algorithmic}[1]\label{alg:OPM}
\STATE Initialize $A = \emptyset$.
\FOR{$i=1,2,\ldots,n$}
\FORALL{$x \in J_i$}
\STATE Set price $p_x = 
\begin{cases}
\phi_x^{-1}(T(A,x)) & \mbox{if $A \cup \{x\} \in \isets$} \\
\infty & \mbox{otherwise.}
\end{cases}$
\ENDFOR
\STATE Post price vector $(p_x)_{x \in J_i}$.
\STATE Bidder $i$ chooses an element $x \in J_i$ (or nothing) 
at these posted prices.
\IF{$x$ is chosen}
\STATE Allocate $x$ to bidder $i$ and charge price $p_x$.
\STATE $A \gets A \cup \{x\}$
\ELSE
\STATE Allocate nothing to bidder $i$ and charge price 0.
\ENDIF
\ENDFOR
\end{algorithmic}
\end{algorithm}

To define mechanism $\mcopies$, we first define an online 
weight-adaptive adversary and then run the prophet inequality
algorithm on the input sequence presented by this adversary,
using its thresholds to define posted prices exactly as in 
mechanism $\mech$ above.  The adversary is designed to 
minimize the mechanism's revenue, subject to the constraint
that the elements are presented in an order that runs 
through all of the elements of $J_1$, then the elements
of $J_2$, and so on.  In fact, it is easy to compute this
worst-case ordering by backward induction, which yields
a dynamic program presented in pseudocode 
as Algorithm~\ref{alg:adversary}.  The dynamic programming
table consists of entries 
$V(A,i)$ denoting the expected revenue that 
$\mcopies$ will gain from selling elements of
the set $J_{i+1} \cup \cdots \cup J_n$, given that
it has already allocated the elements of $A$.
Computing and storing these values requires
exponential time and space, but we are not
concerned with making $\mcopies$ into a
computationally efficient mechanism because its role 
in this paper is merely to provide an intermediate
step in the analysis of mechanism $\mech$.

The formula for $V(A,i)$ is guided by the following
considerations.  Since $\mcopies$ will post prices
$p_x=\phi_x^{-1}(T(A,x))$ for all $x \in J_{i+1}$ 
given that it has already 
allocated $A$, it will not allocate any
element of $J_{i+1}$ if $v_x < p_x$ for all
$x \in J_{i+1}$, and otherwise it will
allocate some element  $x \in J_{i+1}$.
In the former case, its expected revenue from the remaining
elements will be $V(A,i+1)$.  In the latter case,
it extracts revenue $p_x$ from bidder $i+1$ and
expected revenue
$V(A \cup \{x\},i+1)$ from the remaining bidders.
Thus, an adversary who wishes the minimize the
revenue obtained by the mechanism will order the
elements  $x \in J_{i+1}$ in increasing order of
$p_x + V(A \cup \{x\},i+1)$.  Denoting the elements
of $J_{i+1}$ in this order by $x_1, x_2, \ldots, x_k$,
we obtain the formula
\begin{equation} \label{eq:vai}
V(A,i) = \left( \prod_{j=1}^k F_{x_j}(p_{x_j}) \right) \cdot V(A,i+1) 
\;\; + \;\;
\sum_{\ell=1}^k \left( \prod_{j = 1}^{\ell-1} F_{x_j}(p_{x_j}) \right)
\cdot (1 - F_{x_\ell}(p_{x_\ell})) \cdot (p_{x_\ell} + V(A \cup \{x_\ell\}, \, i+1)).
\end{equation}
The first term on the right side accounts for the possibility that bidder $i+1$
buys nothing, while the sum accounts for the possibility that 
bidder $i+1$ buys $x_\ell$, for each $\ell=1,\ldots,k$.

\begin{algorithm}[tb]
\caption{Online weight-adaptive adversary for $\icopies$}
\begin{algorithmic}[1]\label{alg:adversary}
\FOR[Preprocessing loop: fill in dynamic programming table]{$i=n,n-1,\ldots,1$}
\FORALL{feasible sets $A \subseteq J_1 \cup \cdots \cup J_i$}
\IF{$i=n$}
\STATE $V(A,i) = 0$
\ELSE
\STATE $p_x = \phi_x^{-1}(T(A,x))$ for all $x \in J_{i+1}$.
\STATE Sort $J_{i+1}$ in order of increasing $p_x + V(A \cup \{x\},i+1)$.
\STATE Denote this sorted list by $x_1,\ldots,x_k$.
\STATE Compute $V(A,i)$ using formula~\eqref{eq:vai}.
\ENDIF
\ENDFOR
\ENDFOR
\STATE \COMMENT{Main loop: choose the ordering of each set $J_i$}
\STATE Initialize $A = \emptyset$.
\FOR{$i=1,\ldots,n$}
\STATE $p_x = \phi_x^{-1}(T(A,x))$ for all $x \in J_i$.
\STATE Sort the elements of $J_i$ in order of increasing 
$p_x + V(A \cup \{x\},i-1)$.
\STATE Present the elements of $J_i$ to the online algorithm 
in this order.
\IF{$\exists x \in J_i \mbox{ s.t. } v_x \geq p_x$}
\STATE Find the first such $x$ in the ordering of $J_i$, and insert $x$ into $A$.
\ENDIF
\ENDFOR
\end{algorithmic}
\end{algorithm}

\begin{algorithm}[tb]
\caption{Mechanism $\mcopies$ for single-parameter domain $\icopies$.}
\begin{algorithmic}[1]\label{alg:mcopies}
\STATE \COMMENT{Set prices using adversary
coupled with online algorithm}
\STATE Obtain bids $b_x$ for all bidders $x \in \gset$.
\STATE Run Algorithm~\ref{alg:adversary}, using $v_x=b_x$ for all $x$,
to obtain an ordering of $\gset$.
\STATE Set $w(x) = \phi_x(b_x)$ for all $x \in \gset$.
\STATE Present the pairs $(x,w(x))$ to the prophet inequality
algorithm, in the order computed above.
\STATE Obtain thresholds $T(A,x)$ from the prophet inequality algorithm.
\STATE Set price $p_x = \phi_x^{-1}(T(A,x))$ for all $x \in \gset$.
\STATE \COMMENT{Determine allocation and payments}
\STATE Initialize $A = \emptyset$
\FOR{$i=1,\ldots,n$}
\FORALL{$x \in J_i$}
\IF{$b_x \geq p_x$ and $b_y < p_y$ for all $y \in J_i$ that precede
$x$ in the ordering}
\STATE Add $x$ to the set $A$.
\STATE Charge price $p_x$ to bidder $x$.
\ENDIF
\ENDFOR
\ENDFOR
\end{algorithmic}
\end{algorithm}

Mechanism $\mcopies$ has already been described above, and is
specified by pseudocode in Algorithm~\ref{alg:mcopies}.
We note that $\mcopies$ does not satisfy the definition 
of an OPM in~\cite{CHMS}, since the price $p_x$ for $x \in J_i$
may depend on the bids $b_y$ for $y \in J_1 \cup \cdots \cup J_{i-1}$.
However, it retains a key property of OPMs that make them suitable
for analyzing multi-parameter mechanisms: the prices of elements
of $J_i$ are predetermined before any of the bids in $J_i$ are
revealed.

\begin{theorem} Mechanism $\mech$ for BMUMD settings with independent regular valuations obtains a $2$-approximation to the revenue of the optimal deterministic mechanism for matroid feasibility constraints, and a $(4p-2)$-approximation to the revenue of the optimal deterministic mechanism for feasibility constraints that are the intersection of $p$ matroids.
\label{thm:opm}
\end{theorem}
\begin{proof} 
Both $\mech$ and $\mcopies$ are posted-price (hence, truthful) mechanisms that 
always output a feasible allocation.  To prove that the 
allocation is always feasible, one can argue by contradiction:
if not, there must be a step in which the set $A$ becomes infeasible
through adding an element $x$.  However, in both $\mech$ and $\mcopies$,
we can see that the price $p_x$ is infinite in that
case, while bid $b_x$ is greater than or equal to $p_x$, a
contradiction.

The proof of the approximate revenue guarantee follows the outline
 given by equation~\eqref{eq:opm-outline} above.
As explained earlier, the only two steps in that equation that do not
follow from prior work are the relations 
\begin{align}
\label{eq:opm-1}
R(\mech) & \geq R(\mcopies) \\
\label{eq:opm-2}
\Phi(\mcopies) & \geq \frac{1}{\alpha} \Phi(\optcopies).
\end{align}
To justify the second line, observe that the ``adversary''
(Algorithm~\ref{alg:adversary}) that computes the ordering
of the bids is an online weight-adaptive adversary.  This 
is because the adversary does not need to observe the 
values $v_x \, (x \in J_i)$ in order to sort the elements
of $J_i$ in order of increasing $p_x + V(A \cup \{x\},i-1)$.
Thus, the prophet inequality algorithm running on the input
sequence specified by the adversary achieves an
expected virtual surplus that is at least 
$\frac{1}{\alpha} \Phi(\optcopies)$.  Furthermore,
the set of elements selected by $\mcopies$ is 
exactly the same as the set of elements selected
by the prophet inequality algorithm --- the
criterion $b_x \geq p_x$ is equivalent to the
criterion $w(x) \geq T(A,x)$ because $w(x) = \phi_x(b_x), \,
T(A,x) = \phi_x(p_x)$, and $\phi$ is monotone increasing.
This completes the proof of~\eqref{eq:opm-2}.

To prove~\eqref{eq:opm-1} we use an argument that, in effect,
justifies our claim that Algorithm~\ref{alg:adversary} is a 
worst-case adversary for mechanism $\mcopies$.  Specifically,
for each $i=0,\ldots,n$ and each 
feasible set $A \subseteq J_1 \cup \cdots \cup J_i$,
let $R(\mech,A,i)$ and $R(\mcopies,A,i)$ denote the
expected revenue that $\mech$ (respectively, $\mcopies$)
obtains from selling items in $J_{i+1} \cup \cdots J_n$
conditional on having allocated set $A$ while
processing the bids in $J_1 \cup \cdots \cup J_i$.
(In evaluating the expected revenue of the two mechanisms,
we assume that the bidders are presented to $\mech$ in the
order $i=1,\ldots,n$, and that they are presented to 
$\mcopies$ in the order determined by the adversary,
Algorithm~\ref{alg:adversary}.)
We will prove, by downward induction on $i$, that
$$
\forall i,A \quad R(\mech,A,i) \geq R(\mcopies,A,i) = V(A,i)
$$
and then~\eqref{eq:opm-1} follows by specializing
to $i=0, A=\emptyset$.
When $i=n$, we have $R(\mech,A,i) = R(\mcopies,A,i) = V(A,i) = 0$
so the base case of the induction is trivial.  The 
relation $R(\mcopies,A,i) = V(A,i)$ for $i<n$ 
is justified by the discussion preceding equation~\eqref{eq:vai}.
To prove $R(\mech,A,i) \geq R(\mcopies,A,i)$,
suppose that both mechanisms have
allocated set $A$ while processing the bids 
in $J_1 \cup \cdots \cup J_i$.  Conditional
on the set of $x \in J_{i+1}$ such that $v_x \geq p_x$ 
being equal to any specified set $K$,
we will prove that $\mech$ obtains
at least as much expected revenue as $\mcopies$ from 
selling the elements of $J_{i+1} \cup \cdots \cup J_n$.
If $K$ is empty, then the
two mechanisms will obtain expected revenue
$R(\mech,A,i+1)$ and $R(\mcopies,A,i+1)$, respectively,
from elements of $J_{i+1} \cup \cdots \cup J_n$, and
the claim follows from the induction hypothesis.
Otherwise, $\mcopies$ obtains expected revenue
$\min \{ p_x + V(A \cup \{x\},i+1) \mid x \in K \}$
while $\mech$ obtains expected revenue
$p_y + R(\mech,A \cup \{y\},i+1)$ where 
$y \in K$ is the element of $K$
chosen by bidder $i+1$ when presented with the menu of
posted prices for the elements of $J_{i+1}$.  
The induction hypothesis implies 
$$
p_y + R(\mech,A \cup \{y\},i+1) \geq p_y + V(A \cup \{y\},i+1) \geq
\min \{ p_x + V(A \cup \{x\},i+1) \mid x \in K \},
$$
and this completes the proof.
\end{proof}
\bibliographystyle{plain}
\bibliography{mpi}

\begin{thebibliography}{10}

\bibitem{Alaei}
Saeed Alaei.
\newblock Bayesian combinatorial auctions: Expanding single buyer mechanisms to
  many buyers.
\newblock In {\em Proc. 52nd IEEE Symp. on Foundations of Computer Science},
  pages 512--521, 2011.

\bibitem{DBLP:conf/icalp/2009-2}
Susanne Albers, Alberto Marchetti-Spaccamela, Yossi Matias, Sotiris~E.
  Nikoletseas, and Wolfgang Thomas, editors.
\newblock {\em Automata, Languages and Programming, 36th Internatilonal
  Collogquium, ICALP 2009, Rhodes, greece, July 5-12, 2009, Proceedings, Part
  II}, volume 5556 of {\em Lecture Notes in Computer Science}. Springer, 2009.

\bibitem{BIK}
Moshe Babaioff, Nicole Immorlica, and Robert Kleinberg.
\newblock Matroids, secretary problems, and online mechanisms.
\newblock In {\em Proc. 18th ACM Symp. on Discrete Algorithms}, pages 434--443,
  2007.

\bibitem{soda12}
Sourav Chakraborty and Oded Lachish.
\newblock Improved competitive ratio for the matroid secretary problem.
\newblock In {\em SODA12}.
\newblock to appear.

\bibitem{CHMS}
Shuchi Chawla, Jason~D. Hartline, David~L. Malec, and Balasubramanian Sivan.
\newblock Multi-parameter mechanism design and sequential posted pricing.
\newblock In {\em Proc. 41th ACM Symp. on Theory of Computing}, pages 311--320,
  2010.

\bibitem{dgv}
Brian~C. Dean, Michel~X. Goemans, and Jan Vondr{\'a}k.
\newblock Approximating the stochastic knapsack problem: The benefit of
  adaptivity.
\newblock In {\em FOCS}, pages 208--217. IEEE Computer Society, 2004.

\bibitem{dimitrov}
Nedialko~B. Dimitrov and C.~Greg Plaxton.
\newblock Competitive weighted matching in transversal matroids.
\newblock In Luca Aceto, Ivan Damg{\aa}rd, Leslie~Ann Goldberg, Magn{\'u}s~M.
  Halld{\'o}rsson, Anna Ing{\'o}lfsd{\'o}ttir, and Igor Walukiewicz, editors,
  {\em ICALP (1)}, volume 5125 of {\em Lecture Notes in Computer Science},
  pages 397--408. Springer, 2008.

\bibitem{goel}
Ashish Goel, Sanjeev Khanna, and Brad Null.
\newblock The ratio index for budgeted learning, with applications.
\newblock In Claire Mathieu, editor, {\em SODA}, pages 18--27. SIAM, 2009.

\bibitem{gm-switching}
Sudipto Guha and Kamesh Munagala.
\newblock Multi-armed bandits with metric switching costs.
\newblock In Albers et~al. \cite{DBLP:conf/icalp/2009-2}, pages 496--507.

\bibitem{gm-restless}
Sudipto Guha, Kamesh Munagala, and Peng Shi.
\newblock Approximation algorithms for restless bandit problems.
\newblock {\em J. ACM}, 58(1):3, 2010.

\bibitem{HKS-aaai}
MohammadTaghi Hajiaghayi, Robert Kleinberg, and Tuomas~W. Sandholm.
\newblock Automated mechanism design and prophet inequalities.
\newblock In {\em Proc. 22nd AAAI Conference on Artificial Intelligence}, pages
  58--65, 2007.

\bibitem{kennedy85}
D.~P. Kennedy.
\newblock Optimal stopping of independent random variables and maximization
  prophets.
\newblock {\em Ann. Prob.}, 13:566--571, 1985.

\bibitem{kennedy87}
D.~P. Kennedy.
\newblock Prophet-type inequalities for multi-choice optimal stopping.
\newblock {\em Stoch. Proc. Appl.}, 24:77--88, 1987.

\bibitem{kertz86}
R.~P. Kertz.
\newblock Comparison of optimal value and constrained maxima expectations for
  independent random variables.
\newblock {\em Adv. Appl. Prob.}, 18:311--340, 1986.

\bibitem{korula}
Nitish Korula and Martin P{\'a}l.
\newblock Algorithms for secretary problems on graphs and hypergraphs.
\newblock In Albers et~al. \cite{DBLP:conf/icalp/2009-2}, pages 508--520.

\bibitem{KS77}
Ulrich Krengel and Louis Sucheston.
\newblock Semiamarts and finite values.
\newblock {\em Bull. Amer. Math. Soc.}, 83:745--747, 1977.

\bibitem{KS78}
Ulrich Krengel and Louis Sucheston.
\newblock On semiamarts, amarts, and processes with finite value.
\newblock {\em Adv. in Prob. Related Topics}, 4:197--266, 1978.

\bibitem{Mye81}
Roger~B. Myerson.
\newblock {Optimal Auction Design}.
\newblock {\em Mathematics of Operations Research}, 6:58--73, 1981.

\bibitem{S-C}
Ester Samuel-Cahn.
\newblock Comparison of threshold stop rules and maximum for independent
  nonnegative random variables.
\newblock {\em Annals of Probability}, 12(4):1213--1216, 1984.

\bibitem{schrijver-b}
Alexander Schrijver.
\newblock {\em Combinatorial Optimization}, volume~B.
\newblock Springer, 2003.

\bibitem{soto}
Jos{\'e}~A. Soto.
\newblock Matroid secretary problem in the random assignment model.
\newblock In {\em SODA11}, pages 1275--1284, 2011.

\end{thebibliography}

\newpage
\appendix

\section{Proof of Proposition~\ref{prop:alphabetaintersection}}
\label{sec:alphabetaintersection-proof}

\begin{proof}
We have
\begin{equation} \label{eq:optintersections}
\opt = \expect{\wt'(C_j(A)) + \wt'(R_j(A))}\ \forall j
\end{equation}
\begin{equation} \label{eq:optintersections2}
\opt = \expect{\wt'(C(A)) + \wt'(R(A))}
\end{equation}
because $C_j(A) \cup R_j(A)$ is a maximum-weight independent set with respect
to $\wt'$ for all $j$, as is $C(A) \cup R(A)$, and $\wt'$ has the same distribution as $\wt$.
The proof will again consist of deriving the following three inequalities.
\begin{align} \label{eq:step1intersection}
\expect{\sum_{x_i \in A} T_i} & \geq \frac{1}{\alpha} \expect{\sum_j\wt'(C_j(A))} \\ 
\label{eq:step2intersection}
\expect{\sum_{x_i \in A} (\wt_i - T_i)^+} & \geq 
\expect{\sum_{x_i \in R(A)} (\wt'(x_i) - T_i)^+} \\
\label{eq:step3intersection}
\expect{\sum_{x_i \in R(A)} (\wt'(x_i) - T_i)^+} & \geq
\expect{\wt'(R(A))} - \frac{1}{\alpha} \expect{\sum_j\wt'(R_j(A))}.
\end{align}

Summing $\eqref{eq:step1intersection} + \eqref{eq:step2intersection} + \frac{1}{\alpha-1}\eqref{eq:step3intersection}$ and using the fact that $T_i + (\wt_i-T_i)^+ = \wt_i$ for all $x_i \in A$, we obtain
\begin{align*}
\expect{\wt(A)} &\geq \left(\frac{1}{\alpha-1} - \frac{1}{\alpha(\alpha-1)}\right)
 \expect{\sum_j\wt'(C_j(A))} + \frac{\alpha-2}{\alpha-1}\expect{\sum_{x_i \in R(A)} (\wt'(x_i) - T_i)^+}\\
&+ \frac{1}{\alpha-1}\expect{\wt'(R(A))}
- \frac{1}{\alpha(\alpha-1)} \expect{\sum_j \wt'(R_j(A))}.
\end{align*}

Subsituting in Equations \eqref{eq:optintersections} and \eqref{eq:optintersections2} (and observing that $\frac{\alpha-2}{\alpha-1} \geq 0$ whenever $\alpha \geq 2$), we get:

$$\expect{\wt(A)} \geq \frac{1}{\alpha-1}\opt - \frac{p}{\alpha(\alpha-1)}\opt = \frac{\alpha - p}{\alpha(\alpha-1)}\opt$$

It remains to show that Equations \eqref{eq:step1intersection} - \eqref{eq:step3intersection} hold for any $\alpha$-balanced thresholds. Equation \eqref{eq:step1intersection} is again a restatement of the definition of $\alpha$-balanced thresholds. Inequality~\eqref{eq:step2intersection} is deduced from the same observations as Equation \eqref{eq:step2}. Finally, as in Proposition \ref{prop:alphabeta}, we
apply Property~\eqref{eq:betaintersection} of $\alpha$-balanced
thresholds,
using the set $V = R(A)$, to deduce that
\begin{align*}
\expect{\sum_{x_i \in R(A)} \wt'(x_i)} & \leq
\expect{\sum_{x_i \in R(A)} T_i} + 
\expect{\sum_{x_i \in R(A)} (\wt'(x_i) - T_i)^+}  \\
& \leq
\frac{1}{\alpha} \, \expect{\sum_j \wt'(R_j(A))}
+
\expect{\sum_{x_i \in R(A)} (\wt'(x_i) - T_i)^+}  \\
\expect{\wt'(R(A))} - \frac{1}{\alpha} \expect{\sum_j \wt'(R_j(A))} & \leq
\expect{\sum_{x_i \in R(A)} (\wt'(x_i) - T_i)^+} 
\end{align*}
Consequently~\eqref{eq:step3intersection} holds, which concludes the proof.

\end{proof}

\end{document}